%% file: main.tex
\documentclass[11pt,letterpaper]{article}
\usepackage[margin=1in]{geometry}
\usepackage[T1]{fontenc}
\usepackage[utf8]{inputenc}
\usepackage{lmodern}
\usepackage{xcolor}
\usepackage[english]{babel}
\usepackage{tikz,amsmath,amssymb,amsfonts,latexsym,graphicx,algorithmicx}
\usepackage{subcaption}
\usepackage{amsthm}
\usepackage{bbm}
\usepackage{mathtools}
\usepackage{xspace}
 \usepackage[ruled,vlined,linesnumbered]{algorithm2e}
 \usepackage[margin=1in]{geometry}
\usepackage[T1]{fontenc}
\usepackage[utf8]{inputenc}
\usepackage{lmodern}
\usepackage{xcolor}
\usepackage[english]{babel}
\usepackage{tikz,amsmath,amssymb,amsfonts,latexsym,graphicx,algorithmicx}
\usepackage{subcaption}
\usepackage{amsthm}
\usepackage{bbm}
\usepackage{mathtools}
\usepackage{xspace}
 \usepackage[ruled,vlined,linesnumbered]{algorithm2e}

\usepackage[noend]{algpseudocode}
\usepackage{hyperref}
\usepackage[capitalise]{cleveref}

\usepackage{tikz,amsthm,algorithmicx}

\usepackage{mathtools}
 \usepackage[ruled,vlined,linesnumbered]{algorithm2e}

\usepackage[noend]{algpseudocode}
\newtheorem{theorem}{Theorem}
\newtheorem{lemma}[theorem]{Lemma}

\newtheorem{definition}[theorem]{Definition}
\newtheorem{fact}[theorem]{Fact}

\newtheorem*{remark}{Remark}

\usepackage{bm}

\newcommand{\eps}{\varepsilon}

\newcommand{\MM}{\mathcal{M}}

\newcommand{\RR}{\mathcal{R}}

\newcommand{\powerd}{\omega}

\newcommand{\opt}{\mathrm{opt}}

\title{Faster Approximation Scheme for Euclidean $k$-TSP}
\author{
Ernest van Wijland\thanks{École Normale Supérieure de Paris, France.}\and Hang Zhou\thanks{École Polytechnique, IP Paris, France.}}
\date{}

\begin{document}

\maketitle

\begin{abstract}
In the Euclidean $k$-traveling salesman problem ($k$-TSP), we are given $n$ points in the $d$-dimensional Euclidean space, for some fixed constant $d\geq 2$, and a positive integer $k$. The goal is to find a shortest tour visiting at least $k$ points.

We give an approximation scheme for the Euclidean $k$-TSP in time $n\cdot 2^{O(1/\varepsilon^{d-1})} \cdot(\log n)^{2d^2\cdot 2^d}$. This improves Arora's approximation scheme of running time $n\cdot k\cdot (\log n)^{\left(O\left(\sqrt{d}/\varepsilon\right)\right)^{d-1}}$ [J. ACM 1998]. Our algorithm is Gap-ETH tight and can be derandomized by increasing the running time by a factor $O(n^d)$.
\end{abstract}

\thispagestyle{empty}
\setcounter{page}{0}

\newpage
\pagenumbering{arabic}

\setcounter{page}{1}




\section{Introduction}
In the {Euclidean $k$-traveling salesman problem ($k$-TSP)}, we are given $n$ points in $\mathbb{R}^d$, for some fixed constant $d\geq 2$, and a positive integer $k\leq n$.
The goal is to find a shortest tour visiting at least $k$ points out of the $n$ points.

The Euclidean $k$-TSP is NP-hard~\cite{garey1976some}, so researchers turned to approximation algorithms, e.g.,~\cite{arora1998polynomial,awerbuch1995improved,mata1995approximation,mitchell1998constant}.\footnote{$k$-TSP has also been referred as ``quota TSP'' in the literature.}
The best-to-date approximation for the Euclidean $k$-TSP is due to the approximation scheme\footnote{An \emph{approximation scheme} is a $(1+\eps)$-approximation algorithm for any $\eps>0$.} of Arora~\cite{arora1998polynomial}, which is among the most prominent results in combinatorial optimization.
The randomized version of Arora's approximation scheme has a running time of \[n\cdot k\cdot (\log n)^{\left(O\left(\sqrt{d}/\eps\right)\right)^{d-1}}.\]

In this work, we give a faster approximation scheme for the Euclidean $k$-TSP; see \cref{thm:main}.
Compared  with Arora~\cite{arora1998polynomial}, our running time sheds the factor $k$ and, in addition, achieves an asymptotically optimal dependence on $\eps$.

\begin{theorem}
\label{thm:main}
    Let $d\geq 2$ be a fixed constant.
    For any $\eps>0$, there is a randomized $(1+\eps)$-approximation algorithm for the Euclidean $k$-TSP that runs in time \[n\cdot 2^{O(1/\eps^{d-1})} \cdot(\log n)^{2d^2\cdot 2^d}.\]
    The dependence on $\eps$ in the running time is asymptotically optimal under the \emph{Gap-Exponential Time Hypothesis (Gap-ETH)}.
    The algorithm can be derandomized by increasing the running time by a factor $O(n^d)$.
\end{theorem}

In the rest of the section, we outline the proof of \cref{thm:main}.

First, in a preprocessing step, the instance is partitioned into \emph{well-rounded} subinstances.
The partition algorithm in \cite{arora1998polynomial} takes $O(n\cdot k\cdot \log n)$ time.
We improve the running time of the partition algorithm to $O(n)$; see the partition theorem (\cref{thm:partition}).
To that end, we use a result on the \emph{enclosing circles} due to Har-Peled and Raiche~\cite{har2015net}.
See~\cref{sec:partition}.

Next, each subinstance is solved independently using a dynamic program based on the \emph{quadtree}~\cite{arora1998polynomial}.
The dynamic program in \cite{arora1998polynomial} takes $n\cdot k\cdot (\log n)^{\left(O\left(\sqrt{d}/\eps\right)\right)^{d-1}}$ time. We improve the running time of the dynamic program to $n\cdot 2^{O(1/\eps^{d-1})} \cdot(\log n)^{2d^2\cdot 2^{d}}$; see the dynamic programming theorem~(\cref{thm:DP}).
In order to improve the dependence in $\eps$ in the running time in \cite{arora1998polynomial}, we exploit a structure theorem (\cref{thm:structure}) that is a corollary of the approaches of Arora~\cite{arora1998polynomial} and Kisfaludi-Bak, Nederlof, and Węgrzycki~\cite{kisfaludi2022gap}; see \cref{sec:quad}. 
In order to remove the factor $k$ in the running time in \cite{arora1998polynomial}, we discretize the possible lengths of a tour into values called \emph{budgets}; see \cref{sec:budget}.
This is inspired by Kolliopoulos and Rao in the context of \emph{$k$-median}~\cite{kolliopoulos2007nearly}.
The combination of the structure theorem and the budgets is non-trivial and is the key to the improved running time of the dynamic program, see \cref{sec:first-algo,sec:improved-algo}.

The overall running time and the derandomization in \cref{thm:main} follow from the partition theorem~(\cref{thm:partition}) and the dynamic programming theorem~(\cref{thm:DP}).
The Gap-ETH tightness in \cref{thm:main} is a corollary of the Gap-ETH lower bound for the Euclidean TSP~\cite[Theorem~I.1]{kisfaludi2022gap}.

%

\section{Notations}

Let $P$ denote a set of $n$ points in $\mathbb{N}^d$ for some fixed constant $d\geq 2$.
Let $k$ be an integer in $[1,n]$.
A path $\pi$ in $\mathbb{N}^d$ is a \emph{$k$-salesman tour} if $\pi$ is a closed path visiting at least $k$ points from $P$.
Let $w(\pi)$ denote the \emph{length} of $\pi$.
In the \emph{Euclidean $k$-traveling salesman problem ($k$-TSP)}, we look for a $k$-salesman tour $\pi$ that minimizes $w(\pi)$.
Let $\opt$ denote the minimum length of a $k$-salesman tour.
For notational convenience, let $\powerd$ denote $2^d$.

\begin{definition}[well-rounded instance, {\cite[Section~3.2]{arora1998polynomial}}]
\label{def:well-rounded}
Consider an instance for the Euclidean $k$-TSP.
Let $L\in\mathbb{N}$ denote the side length of the bounding box for the instance.
We say that the instance is \emph{well-rounded} if $L=O(k^2)$, all points in the instance have integral coordinates in $\{0,\dots,L\}^d$, and the minimum nonzero internode distance is at least $8$.
\end{definition}

\section{Partitioning Into Subinstances}

\label{sec:partition}

\begin{theorem}[partition theorem]
\label{thm:partition}
    Let $\mathcal{I}$ be an instance for the Euclidean $k$-TSP.
    There is a randomized algorithm that computes in expected $O(n)$ time a partition of $\mathcal{I}$ into a family of well-rounded subinstances $\mathcal{I}_1,\dots,\mathcal{I}_\ell$ for some $\ell\geq 1$, such that with probability at least $1-2/\log k$, an optimal solution to $\mathcal{I}$ is completely within $\mathcal{I}_j$ for some $j\in [1,\ell]$.
    The algorithm can be derandomized by increasing the running time by a factor $n$.
\end{theorem}

In the rest of the section, we prove \cref{thm:partition}.

Recall that, in the approach of Arora~\cite{arora1998polynomial}, an important step is to compute a good approximation for the optimal cost.
\cref{fact:smallest-square} relates the optimal cost with the side length of the smallest $d$-dimensional hypercube containing $k$ points.

\begin{fact}[{\cite[Section~3.2]{arora1998polynomial}}]
\label{fact:smallest-square}
The cost of the optimal solution to the Euclidean $k$-TSP is at most $dk^{1-(1/d)}$ times larger than the side length of the smallest $d$-dimensional hypercube containing $k$ points.
\end{fact}

Estimating the smallest $d$-dimensional hypercube containing $k$ points takes $O(nk\log n)$ time in \cite{arora1998polynomial}.
We improve this running time to $O(n)$ in \cref{lem:har-cube}.
This is achieved using a result of Har-Peled and Raiche~\cite{har2015net} on the estimation of the smallest $d$-dimensional \emph{ball} (\cref{lem:har-ball}).

\begin{lemma}[{\cite[Corollary 4.18]{har2015net}}]
\label{lem:har-ball}
Let $\lambda>0$.
For a set of $n$ points in $\mathbb{R}^d$ and a positive integer $k$, one can $(1+\lambda)$-approximate, in expected $O(n/\lambda^d)$ time, the radius of the smallest $d$-dimensional ball containing $k$ points.
\end{lemma}

\begin{lemma}
\label{lem:har-cube}
For a set of $n$ points in $\mathbb{R}^d$ and a positive integer $k$, one can $2\sqrt{d}$-approximate, in expected $O(n)$ time, the side length of the smallest $d$-dimensional hypercube containing $k$ points.
\end{lemma}

\begin{proof}
Let $R^*$ denote the radius of the smallest $d$-dimensional ball containing $k$ points.
Applying \cref{lem:har-ball} with $\lambda=1$, one can compute in expected $O(n)$ time an estimate $R$ such that $R^*\leq R\leq 2R^*$.
On the one hand, there exists a $d$-dimensional hypercube of side length $2R^*\leq 2R$ that contains at least $k$ points.
On the other hand, since the longest diagonal of a $d$-dimensional hypercube is equal to $\sqrt{d}$ times the side length of that hypercube, any $d$-dimensional hypercube of side length $2R^*/\sqrt{d}\geq R/\sqrt{d}$ contains at most $k$ points.
So the side length of the smallest $d$-dimensional hypercube containing $k$ points is in $[R/\sqrt{d},2R]$.
The claim follows.
\end{proof}

\begin{proof}[Proof of the partition theorem (\cref{thm:partition})]
From \cref{fact:smallest-square} and \cref{lem:har-cube}, there is a randomized algorithm that computes in expected $O(n)$ time an estimate $A$ for the cost of the optimal solution such that \[\opt\leq A\leq 2d^{3/2}k^{1-(1/d)}\cdot \opt.\]
The first part of the claim follows from arguments that are almost identical to \cite{arora1998polynomial}, except by modifying the definition of the parameter $\rho$ to \[\rho:=\frac{A\eps}{16d^{3/2}k^{2-(1/d)}}.\]

Now we prove the second part of the claim.
We only need to derandomize the algorithm in \cref{lem:har-cube}.
This requires derandomizing the algorithm in \cref{lem:har-ball}, which is called $\texttt{ndpAlg}$ in \cite[Figure 3.1]{har2015net}.
To that end, we remove Line $1$ of the algorithm \texttt{ndpAlg} in \cite{har2015net}, which randomly picks a point $p$ from the set $W_{i-1}$.
Instead, we enumerate \emph{all} points $p$ from $W_{i-1}$.
For each such point $p$, we compute a set $W_i^p$ using Lines 2--7 of the algorithm \texttt{ndpAlg} in \cite{har2015net}.
Finally, we let $W_i$ be the set $W_i^p$ with minimum cardinality for all points $p\in W_{i-1}$.
This completes the description of the derandomized algorithm, see \cref{alg:partition}.

\begin{algorithm}
   \caption{Derandomization for \texttt{ndpAlg}.  $W\subseteq \mathbb{R}^d$ denotes a set of $n$ input points.}
   \label{alg:partition}
   $W_0\gets W$\\
   $i\gets 1$\\
   \While{$W_{i-1}\neq \emptyset$}{
       \ForAll{$p\in W_{i-1}$}{
           Compute $W_i^p$ \Comment{Lines 2-7 of algorithm \texttt{ndpAlg} in \cite{har2015net}}
       }
       Let $W_i$ be the set $W_i^p$ with minimum cardinality for all $p\in W_{i-1}$\\
       $i\gets i+1$
   }
\end{algorithm}

It remains to show that the running time of the derandomized algorithm (\cref{alg:partition}) is $O(n^2)$.
Let $z$ denote the number of iterations in the \textbf{while} loop in \cref{alg:partition}.
Consider any integer $i\in [1,z]$.
From the analysis in \cite[Lemma~3.12]{har2015net}, for each $p\in W_{i-1}$, the set $W_i^p$ in \cref{alg:partition} can be computed in $O(|W_{i-1}|)$ time.
So the overall time to compute $W_i^p$ over all $p\in W_{i-1}$ is $O(|W_{i-1}|^2)$.
Using an analysis similar to \cite[Lemma~3.12]{har2015net}, there exists some $p\in W_{i-1}$ such that
$|W_i^p|\leq (15/16)|W_{i-1}|$, thus
$|W_i|\leq (15/16)|W_{i-1}|$ by the definition of $W_i$.
Therefore, the overall running time of \cref{alg:partition} is
\[\sum_{i=1}^{z} O(|W_{i-1}|^2)\leq \sum_{i=1}^{z} (15/16)^{i-1}\cdot O(|W_0|^2)=O(|W_0|^2)=O(n^2).\]
This completes the proof of the second part of the claim.
\end{proof} 

\section{Dynamic Programming}

\label{sec:DP}
\begin{theorem}[dynamic programming theorem]
\label{thm:DP}
    Consider a well-rounded instance for the Euclidean $k$-TSP. 
    There is a randomized  algorithm with running time $n\cdot 2^{O(1/\eps^{d-1})} \cdot(\log n)^{2d^2\cdot 2^{d}}$ such that, with probability at least $1/2$, the algorithm outputs a $(1+\eps)$-approximate solution. 
    The algorithm can be derandomized by increasing the running time by a factor $O(n^d)$.
\end{theorem}

In the rest of the section, we prove \cref{thm:DP}.

\subsection{Preliminaries: Notations, Quadtree, and Structure Properties}
\label{sec:quad}

Let $L$ denote the side length of the bounding box of the instance. 
Since the instance is well-rounded (\cref{def:well-rounded}), we may assume that $P\subseteq \{0,\dots,L\}^d$ and $L=O(k^2)$.

We review the \emph{quadtree}~\cite{arora1998polynomial} as well as its structural properties established by Arora~\cite{arora1998polynomial} and 
Kisfaludi-Bak, Nederlof, and Węgrzycki~\cite{kisfaludi2022gap}.

We follow the notations in \cite{kisfaludi2022gap}.
We pick $a_1,\dots,a_d\in \{1,\dots,L\}$ independently and uniformly at random and define $\bm a := (a_1,\ldots, a_d)\in\{0,\dots,L\}^d$.
Consider the hypercube
\[C(\bm a):= \bigtimes\limits_{i=1}^d [-a_i+1/2, 2L-a_i+1/2].\]
Note that $C(\bm a)$ has side length $2L$ and each point from $P$ is contained in $C(\bm a)$.

We define the {\it dissection} of $C(\bm a)$ to be a tree constructed recursively, where each vertex is associated with a hypercube in $\mathbb R^d$.
The root of the tree is associated with $C(\bm a)$.
Each non-leaf vertex of the tree that is associated with a hypercube $\bigtimes_{i=1}^d [l_i,u_i]$ has $\powerd$ children with which we associate hypercubes $\bigtimes_{i=1}^d I_i$, where $I_i$ is either $[l_i,(l_i+u_i)/2]$ or $[(l_i+u_i)/2,u_i]$. Each leaf vertex of the tree is associated with a hypercube of unit length.

A \emph{quadtree} is defined similarly as the dissection of $C(\bm a)$, except we stop the recursive partitioning as soon as the associated hypercube of a vertex contains at most one point from $P$.
Each hypercube associated with a vertex in the quadtree is called a {\it cell} in the quadtree.

For each cell $C$ in the quadtree, let $\partial C$ denote the union of all facets of $C$.

\begin{definition}[grid, {\cite[Definition~II.4]{kisfaludi2022gap}}]
\label{def:grid}
Let $F$ be a $(d-1)$-dimensional hypercube.
Let $t$ be a positive integer.
We define \emph{grid$(F,t)$}$\subseteq\mathbb{R}^{d-1}$ to be an orthogonal lattice of $t$ points in $F$.
Thus, if the hypercube has side length $l$, the minimum distance between any pair of points of \emph{grid$(F, t)$} is $l/t^{1/(d-1)}$.
\end{definition}


\begin{definition}[fine multiset, adaptation from {\cite[Section 5.1 in the full version]{kisfaludi2022gap}}]
\label{def:fine}
Let $m$ and $r$ be positive integers.
Let $C$ be a cell in the quadtree.
Let $B$ be a multiset of points in $\partial C$.
For each facet $F$ of $C$, let $b_F$ denote the number of points in $B$ that are in $F$.
We say that $B$ is \emph{$(m,r)$-fine} if, for all facets $F$ of $C$, either one of the two following cases holds:
\begin{enumerate}
\item $b_F\leq 1$ and $B\cap F\subseteq\emph{grid}(F,m)$;
\item $b_F\geq 2$ and $B\cap F \subseteq \emph{grid}(F, g(b_F))$, where $g(\cdot)$ is an integer-valued function such that $g(b_F)\leq r^{2d-2}/b_F$. Moreover, each point in $\emph{grid}(F, g(b_F))$ occurs at most twice in $B\cap F$.
\end{enumerate}
The parameters $(m,r)$ are omitted when clear from the context. 
\end{definition}

\begin{definition}[$(m,r)$-simple paths, adaptation from {\cite[Definition~1]{arora1998polynomial}} and {\cite[Definition~III.2]{kisfaludi2022gap}}]
\label{def:r-simple}
Let $m$ and $r$ be positive integers. A collection $\mathcal{Q}$ of paths in $\mathbb{R}^d$ is $(m,r)$-{\normalfont simple} if, for every cell $C$, the intersection between $\mathcal{Q}$ and $\partial C$ is $(m,r)$-fine.
\end{definition}


\begin{theorem}[structure theorem, corollary of \cite{arora1998polynomial} and \cite{kisfaludi2022gap}]
\label{thm:structure}
Let $\bm a$ be a random vector in $\{1,\dots,L\}^d$.
Let $m=(O((\sqrt{d}/\eps)\log L))^{d-1}$ and $r=O(d^2/\eps)$.
With probability at least $1/2$, there is a $k$-salesman tour $\pi$ such that both of the following properties hold:
\begin{itemize}
\item the path collection $\{\pi\}$ is $(m,r)$-simple;
\item $w(\pi)\leq (1+ \eps)\cdot\opt$.
\end{itemize}
\end{theorem}

\begin{proof}
From~Theorem~5 and Theorem~10 in~\cite{arora1998polynomial}, for some $m=(O((\sqrt{d}/\eps)\log L))^{d-1}$, the quadtree defined by $\bm a$ has an associated $k$-salesman tour $\pi_0$ such that\footnote{The bound in expectation is obtained in the proofs in~\cite{arora1998polynomial}.}
\begin{equation}
\label{eqn:w0}
\mathbb{E}[w(\pi_0)]\leq (1+\eps/6)\cdot \opt,
\end{equation}
 and $\pi_0$ crosses each facet $F$ of each cell of the quadtree only at points from $\text{grid}(F,m)$.

We then apply Theorem~III.3 from \cite{kisfaludi2022gap} on $\pi_0$ for some parameter $r\in \mathbb{R}$.
This results in an $(m, r)$-simple tour $\pi$ visiting the same set of points as $\pi_0$, such that 
\begin{equation}
\label{eqn:w}
\mathbb{E}[w(\pi)]\leq (1+O(d^2/r))\cdot w(\pi_0)\leq (1+\eps/6)\cdot w(\pi_0),
\end{equation}
where the second inequality holds for some $r=O(d^2/\eps)$ that is well-chosen.

If $\pi$ crosses a facet $F$ of a cell of the quadtree only once, letting $q$ denote that crossing, then $q$ belongs to $\pi_0$.
Since $\pi_0$ crosses each facet $F$ only at points from $\text{grid}(F,m)$, the above crossing $q$ belongs to $\text{grid}(F,m)$.

From \eqref{eqn:w0} and \eqref{eqn:w}, we have 
\[\mathbb{E}[w(\pi)]\leq (1+\eps/6)^2\cdot\opt<(1+ \eps/2)\cdot\opt.\]
Markov's inequality implies that, with probability at least $1/2$, we have $w(\pi_0)\leq (1+\eps)\cdot \opt$.
This completes the proof of the claim.
\end{proof}

\subsection{Budget Multipath Problem}
\label{sec:budget}

\begin{definition}[budgets]
\label{def:budget}
Let $\Phi=dk^{1-1/d}L$. We say that $s\in\mathbb{R}$ is a \emph{budget} if either 
\begin{itemize}
    \item $s=0$; or 
    \item $s\in[1/(r^2+m^{1/(d-1)}), (1+\eps)^2 \cdot \Phi]$ and there exists $i\in \mathbb{N}$ such that \[(1+\eps/(2d\cdot r^{d-1} + 3 \log_2 n))^i=s.\]
\end{itemize}
Let $\mathcal{S}$ be the set of all budgets.
\end{definition}

\begin{definition}[budget multipath problem]
\label{def:BMP}
We are given
\begin{itemize}
\item a cell $C$ in the quadtree,
\item a fine multiset $B\subseteq\partial C$,
\item a perfect matching $M$ on $B$,
\item a budget $s\in\mathcal{S}$.
\end{itemize}
We look for a collection $\mathcal{Q}$ of paths in $C$ satisfying all of the following properties:
\begin{itemize}
    \item $\mathcal{Q}$ is $(m, r)$-simple;
    \item $\mathcal{Q}$ has total length at most $s$;
    \item the intersection between $\mathcal{Q}$ and $\partial C$ is $B$;
    \item there is a one-to-one correspondence between the paths in $\mathcal{Q}$ and the edges in $M$, where we say that a path $q\in\mathcal{Q}$ \emph{corresponds} to an edge $(u, v)\in M$ if and only if $u$ and $v$ are the two endpoints of $q$.
\end{itemize}
The goal is to maximize the number of points in $P$ visited by $\mathcal{Q}$.
\end{definition}

The Euclidean $k$-TSP can be reduced to the budget multipath problem.
To see the reduction, consider the budget multipath problem for the  root cell $C_0$ in the quadtree, the multiset $B:=\emptyset$, the set $M:=\emptyset$, and every budget $s\in\mathcal{S}$. 
Let $s^*\in \mathcal{S}$ be the minimum $s$ such that the solution $\mathcal{Q}$ to the above budget multipath problem on $(C_0,\emptyset,\emptyset,s)$ visits at least $k$ points.
We will show in \cref{sec:improved-algo} that $s^*$ is a near-optimal solution to the Euclidean $k$-TSP.

\subsection{First Algorithm: Dynamic Program with Budgets}
\label{sec:first-algo}

To simplify the presentation, we start by presenting in this section a first algorithm for the budget multipath problem that conveys the main ideas in the algorithmic design, although its running time is not as good as claimed in \cref{thm:DP}.

The algorithm is a dynamic program parameterized by the budget; see \cref{{sec:first-algo-construction}}.
The analysis of the algorithm contains the main technical novelty of the paper; see \cref{sec:first-algo-analysis}.

Later in \cref{sec:improved-algo}, we improve the running time of the algorithm so as to achieve the claimed running time in \cref{thm:DP}.

\subsubsection{Construction}
\label{sec:first-algo-construction}

Consider a fixed cell $C$, a fixed budget $s\in \mathcal{S}$, and a fixed fine multiset $B\subseteq\partial C$.
We construct a set $\MM^C_s[B]$ of pairs $(M,\kappa)$, where $M$ is a perfect matching on $B$ and $\kappa$ is an integer.
Intuitively, $\kappa$ indicates the number of points that can be visited by a collection of paths $\mathcal{Q}$ such that there is a one-to-one correspondence between the paths in $\mathcal{Q}$ and the edges in $M$ and $\mathcal{Q}$ has total cost at most $s$.

For notional convenience, we denote
\[\MM^C_s:=\bigcup\limits_{\text{fine multiset } B\subseteq\partial C}\MM^C_s[B].\]

We construct $\MM^C_s[B]$ in the bottom up order of the cell $C$ in the quadtree; for a fixed cell $C$, in increasing order of the budget $s\in \mathcal{S}$; and for a fixed budget $s$, in non-decreasing order of cardinality of the fine multiset $B\subseteq \partial C$. 

\paragraph*{Leaf Cells}
Consider a leaf cell $C$.
We construct $\MM^C_s[B]$ in non-decreasing order of $|B|$.

\emph{Case 1: $|B|=0$.}
$\MM^C_s[B] := \{(\emptyset, 0)\}$.

\emph{Case 2: $|B|=2$.}
Let $u$ and $v$ be the two elements in $B$.
Since $C$ is a leaf cell, there are two subcases.

\emph{Subcase 2.1:} $C\cap P=\emptyset$. Let
\[\MM^C_s[\{u,v\}] :=
\begin{cases}
\{(\{(u,v)\}, 0)\} & \text{if } \text{dist}(u,v)\leq s \\
\emptyset & \text{if }\text{dist}(u,v)> s
\end{cases}\]

\emph{Subcase 2.2:} $C\cap P=\{p\}$ for some point $p$, letting $n_p\in\mathbb{N}$ denote the multiplicity of $p$ in $P$.
Let
\[\MM^C_s[\{u,v\}] :=
\begin{cases}
\{(\{(u,v)\}, n_p)\} & \text{if } \text{dist}(u,p)+\text{dist}(p,v)\leq s \\
\{(\{(u,v)\}, 0)\} & \text{if } \text{dist}(u,v)\leq s<\text{dist}(u,p)+\text{dist}(p,v) \\
\emptyset & \text{if }\text{dist}(u,v)> s
\end{cases}\]

\emph{Case 3: $|B|>2$.}
We construct $\MM^C_s[B]$ using the following formula:
\begin{equation}
\MM^C_s[B] := \bigcup\limits_{\substack{s_1+s_2\leq s\\ s_1,s_2\in\mathcal{S}}}\left(\bigcup\limits_{\substack{u, v\in B\\ u\neq v}} \left\{(M\cup\{(u, v)\}, \kappa) | (M, \kappa)\in\MM^C_{s_1}[B\setminus\{u, v\}], \text{dist}(u, v)\leq s_2\right\}\right).
\end{equation}

\paragraph*{Non-Leaf Cells}
Consider a non-leaf cell $C$.
Let $C_1,\ldots,C_{\powerd}$ be the $\powerd$ children of $C$ in the quadtree. First, we enumerate all possible budgets $s_1,\dots,s_{\powerd}$ for the $\powerd$ children,  such that $\sum_i s_i\leq s$.
Next, we enumerate all possible pairs $(M_1,\kappa_1)\in\MM^{C_1}_{s_1},\ldots,(M_{\powerd},\kappa_{\powerd})\in\MM^{C_{\powerd}}_{s_{\powerd}}$.
As in \cite{kisfaludi2022gap}, we say that the matchings $M_1, \dots, M_\powerd$ are \emph{compatible} if (1) for any pair of neighboring cells $C'$ and $C''$, the endpoints of the matchings on a shared facet are the same; and (2) combining $M_1, \dots, M_\powerd$ results in a set of paths with endpoints in $\partial C$. 
If the matchings $M_1,\dots, M_\powerd$ are compatible, we let $M$ denote the matching that is the result of \texttt{Join}$(M_1,\dots,M_\powerd)$, where the \texttt{Join} operation is defined in \cite{kisfaludi2022gap}.
If $B$ equals the multiset consisting of the endpoints of the edges in $M$, we insert the pair $(M,\sum_i\kappa_i)$ into $\MM^C_s[B]$.

\subsubsection{Analysis}
\label{sec:first-algo-analysis}
The following lemma shows that discretizing the possible lengths of a path into budgets in the construction of \cref{sec:first-algo-construction} preserves the near-optimality of the cost of the solution.
Its proof is delicate.

\begin{lemma}
\label{lem:discrete}
Let $C$ be any cell.
Let $\mathcal{Q}$ be an $(m,r)$-simple path collection in $C$ of length at most $(1+\eps)\cdot \Phi$.
Let $\kappa$ denote the number of points visited by $\mathcal{Q}$.
Let multiset $B$ consist of the intersection points between $\mathcal{Q}$ and $\partial C$.
Let $M$ denote a perfect matching on the points in $B$ such that there is a one-to-one correspondence between the paths in $\mathcal{Q}$ and the edges in $M$.
Then there exist $\kappa'\geq \kappa$ and $s\in \mathcal{S}$ such that $(M,\kappa')\in \mathcal{M}^C_s[B]$ and $s$ is at most $(1+\eps)$ times the total length of the paths in $\mathcal{Q}$.
\end{lemma}

In the rest of \cref{sec:first-algo-analysis}, we prove \cref{lem:discrete}.

Let $\tau$ denote the total length of the paths in $\mathcal{Q}$.
We prove the claim in two cases, depending on whether $\tau$ is smaller or greater than 
$1/(r^2+m^{1/(d-1)})$.
 
\paragraph*{Case 1: $\tau<1/(r^2+m^{1/(d-1)})$.}

We show that for any $(u, v)\in M$, $u = v$.
If $B=\emptyset$, this is trivial.
Assume that $B\neq\emptyset$, and let $p_1$ and $p_2$ be two points in $B$.

\begin{fact}
Either $\emph{dist}(p_1, p_2)\geq 1/(r^2+m^{1/(d-1)})$ or $p_1 = p_2$.
\end{fact}

\begin{proof}
Let $l$ be the side length of $C$. By \cref{sec:quad}, $l\geq 1$.

\noindent{\bf Case (a): There exists a facet $F$ of $C$ that contains both $p_1$ and $p_2$.}
Since $\mathcal{Q}$ is an $(m, r)$-simple collection of paths, by the definition of $B$ in the claim and by the definition of $(m, r)$-simple paths (\cref{def:r-simple}), $B$ is fine. Let $b_F$ denote the number of points of $B$ that are in $F$. By the definition of a fine multiset (\cref{def:fine}), either $p_1=p_2$ or $p_1, p_2\in\text{grid}(F, g(b_F))$, where $g(\cdot)$ is an integer-valued function such that $g(b_F)\leq r^{2d-2}/b_F$. Hence either $p_1=p_2$ or, by \cref{def:grid}, $\text{dist}(p_1, p_2)\geq l/(r^{2d-2})^{1/(d-1)}\geq 1/(r^2+m^{1/(d-1)})$.


\noindent{\bf Case (b): There there exists no facet of $C$ that contains both $p_1$ and $p_2$.}
Let $F_1$ and $F_2$ be two facets of $C$, such that $F_1$ contains $p_1$ and $F_2$ contains $p_2$. The non-trivial case is when $F_1$ and $F_2$ are neighboring faces. In this case, neither $p_1$ nor $p_2$ can lie on the intersection, because we are in the case where no facet of $C$ contains both $p_1$ and $p_2$. Since $B$ is a fine multiset, there are two cases according to \cref{def:fine}. In the first case, by \cref{def:grid}, for every facet $F$ of $C$, the minimum distance between two points in $\text{grid}(F, m)$ is $l/m^{1/(d-1)}$. Therefore, for both $i\in\{1,2\}$, we have $\text{dist}(p_i, F_1\cap F_2)\geq l/m^{1/(d-1)}$. Similarly, in the second case, we have $\text{dist}(p_i, F_1\cap F_2)\geq l/(r^{2d-2})^{1/(d-1)}$.
Furthermore, $F_1$ and $F_2$ are orthogonal.
Therefore, 
\[\text{dist}(p_1, p_2)\geq \sqrt{d}\min\{l/m^{1/(d-1)}, l/(r^{2d-2})^{1/(d-1)}\}\geq 1/(r^2+m^{1/(d-1)}).\]
This completes the proof of the claim.
\end{proof}

Hence for any $(u, v)\in M$, either $u=v$ or $\text{dist}(u, v)\geq 1/(r^2+m^{1/(d-1)})$. Since $1/(r^2+m^{1/(d-1)})>\tau\geq \sum_{(u,v)\in M}\text{dist}(u,v)$, it is impossible that there exists $(u, v)\in M$ such that $\text{dist}(u, v)\geq 1/(r^2+m^{1/(d-1)})$. Thus, for any $(u, v)\in M$, $u = v$.

By the definition of  the quadtree in \cref{sec:quad}, the distance from any point in $P$ to $\partial C$ is at least $1/2$, so $\mathcal{Q}$ does not visit any point from $P$. Hence $\kappa=0$. By induction on the size of $M'$, for any perfect matching $M'$ on $B$ such that $\forall (u,v)\in M', u=v$, we have $(M', 0)\in \MM_0^C[B]$. Thus, $(M, \kappa')\in\MM_s^C[B]$, where $s=0\leq (1+\eps)\cdot \tau$ and $\kappa'=0\geq\kappa$. The lemma holds when $\tau<1/(r^2+m^{1/(d-1)})$.

\paragraph*{Case 2: $\tau\geq 1/(r^2+m^{1/(d-1)})$.}

\begin{fact}
\label{fct:finebound}
Let $C$ be a cell in the quadtree.
Let $B$ be a fine multiset of points in $\partial C$. We have $|B|\leq 2d\cdot 2r^{d-1}$.
\end{fact}
\begin{proof}
Let $F$ be a facet of $C$ and $b_F$ denote the number of points of $B$ that are on $F$. By \cref{def:fine}, either one of the following two cases holds: (i) $b_F\leq 1$ and $B\cap F\subseteq\text{grid}(F, m)$; (ii) $b_F\geq 2$ and $B\cap F\subseteq\text{grid}(F, g(b_F))$ for some $g(b_F)\leq r^{2d-2}/b_F$. Moreover, each point in $\text{grid}(F, g(b_F))$ occurs at most twice in $B\cap F$.

We show that $b_F\leq 2r^{d-1}$. In case (i), this is trivial. In case (ii), since each point from $\text{grid}(F, g(b_F))$ is contained at most twice in $B\cap F$, we have $b_F\leq 2g(b_F)$. Together with $g(b_F)\leq r^{2d-2}/b_F$, we have $b_F\leq 2r^{d-1}$. Furthermore, since $C$ is a $d$-dimensional hypercube, $C$ has $2d$ facets. 
The claim follows. 
\end{proof}

\begin{lemma}
\label{lem:induc}
Let $C$ be any cell.
Let $h$ denote the height of the subtree rooted at $C$.
Let $\mathcal{Q}$ be an $(m,r)$-simple path collection in $C$ of length at most $(1+\eps)\cdot \Phi$.
Let $\tau$ denote the length of $\mathcal{Q}$.
Let $\kappa$ denote the number of points visited by $\mathcal{Q}$.
Let multiset $B$ consist of the intersection points between $\mathcal{Q}$ and $\partial C$.
Let $M$ denote a perfect matching on the points in $B$ such that there is a one-to-one correspondence between the paths in $\mathcal{Q}$ and the edges in $M$.
Then there exist $\kappa'\geq \kappa$ and $s\in \{\left(1+\eps/(2d\cdot r^{d-1}+3\log_2 n)\right)^i, i\in\mathbb{Z}\}$ such that $(M,\kappa')\in \mathcal{M}^C_s[B]$ and $s\leq \tau\cdot (1+\eps/(2d\cdot r^{d-1} + 3 \log_2 n))^{2d\cdot r^{d-1}+h}$.
\end{lemma}
\begin{proof}
We proceed by induction in the bottom-up order of the cell $C$ in the quadtree.

\noindent{\bf Case (a): $C$ is a leaf cell.}
Observe that $B$ is obtained by $|B|/2$ inclusion operations of pairs of points in the construction in \cref{sec:first-algo-construction}.
By \cref{def:r-simple}, $B$ is a fine multiset. Therefore, by \cref{fct:finebound}, we have $|B|\leq 2d\cdot 2r^{d-1}$.
Since the cost inside $C$ is obtained after at most $|B|/2\leq 2d\cdot r^{d-1}$ rounding operations, there exists $s\in \{\left(1+\eps/(2d\cdot r^{d-1}+3\log_2 n)\right)^i, i\in\mathbb{Z}\}$ such that $s\leq \tau \cdot (1+\eps/(2d\cdot r^{d-1} + 3 \log_2 n))^{2d\cdot r^{d-1}}$ such that $(M,\kappa')\in \mathcal{M}^C_s[B]$ for some $\kappa'\geq \kappa$.

\noindent{\bf Case (b): $C$ is a non-leaf cell.}
Let $C_1,\dots,C_{\powerd}$ be the children of $C$ in the decomposition. Let $\kappa_i$ denote the number of points visited by $\mathcal{Q}$ inside of $C_i$. Let $\tau_i$ denote the cost of $\mathcal{Q}$ inside of $C_i$. Let $h_i$ denote the height of the subtree rooted at $C_i$.
By induction, for each $i\in [1,\powerd]$,
there exist $\kappa_i'\geq \kappa_i$ and 
$s_i\in \{\left(1+\eps/(2d\cdot r^{d-1}+3\log_2 n)\right)^i, i\in\mathbb{Z}\}$ such that $(M_i,\kappa_i')\in \mathcal{M}^{C_i}_{s_i}[B_i]$ and $s_i\leq \tau_i\cdot (1+\eps/(2d\cdot r^{d-1} + 3 \log_2 n))^{2d\cdot r^{d-1}+h_i}$.

Let $\kappa':=\sum_i \kappa_i'$ and let $s$ be the smallest budget in $\mathcal{S}$ that is at least $\sum_i s_i$.

Combining the solutions in all $C_i$ and noting that $h\geq h_i+1$ for all $i$, we have 
\[(M, \kappa')\in \MM_s^C[B]\] and 
\[\kappa'= \sum_i \kappa_i'\geq \sum_i \kappa_i=\kappa\]
and
\begin{align*}
s&\leq (1+\eps/(2d\cdot r^{d-1} + 3 \log_2 n))\cdot \sum_i s_i\\
&\leq (1+\eps/(2d\cdot r^{d-1} + 3 \log_2 n))\cdot \sum_i \tau_i\cdot (1+\eps/(2d\cdot r^{d-1} + 3 \log_2 n))^{2d\cdot r^{d-1}+h_i}\\
&\leq (1+\eps/(2d\cdot r^{d-1} + 3 \log_2 n))^{2d\cdot r^{d-1}+h}\sum_i \tau_i\\
&=(1+\eps/(2d\cdot r^{d-1} + 3 \log_2 n))^{2d\cdot r^{d-1}+h}\cdot\tau.
\end{align*}
This completes the proof of the claim.
\end{proof}

Finally, let us bound $2d\cdot r^{d-1}+h$. 
Since the instance is well-rounded, there exists an absolute constant $D$ such that the size of the bounding box is at most $Dk^2$ (\cref{def:well-rounded}). Therefore, the height of the quadtree is at most $\lceil\log_2(Dk^2)\rceil\leq\log_2(Dn^2)+1\leq \log_2 D + 2\log_2 n + 1\leq 3\log_2 n$. Thus, $2d\cdot r^{d-1}+h\leq 2d\cdot r^{d-1} + 3 \log_2 n$ for $n$ large enough. Hence
\begin{align*}
\tau\cdot (1+\eps/(2d\cdot r^{d-1} + 3 \log_2 n))^{2d\cdot r^{d-1}+h}&\leq \tau\cdot (1+\eps/(2d\cdot r^{d-1} + 3 \log_2 n))^{2d\cdot r^{d-1} + 3 \log_2 n}\\
&\leq \tau\cdot (1+\eps).
\end{align*}
By \cref{lem:induc}, there exists $s\in\{\left(1+\eps/(2d\cdot r^{d-1}+3\log_2 n)\right)^i, i\in\mathbb{Z}\}$ and $\kappa'\geq\kappa$ such that $(M, \kappa')\in\MM_s^C[B]$ and $s\leq (1+\eps/(2d\cdot r^{d-1}+3\log_2 n)^{2d\cdot r^{d-1}+h}\cdot \tau$.
We have $s\leq (1+\eps)\cdot \tau\leq (1+\eps)^2\cdot \Phi$. Therefore, $s\in\mathcal{S}$.
This completes the proof of \cref{lem:discrete}.

\subsection{Improved Algorithm and Proof of \cref{thm:DP}}
\label{sec:improved-algo}
\subsubsection{Construction}
In order to achieve the claimed running time in \cref{thm:DP}, we combine the algorithm in \cref{sec:first-algo} with the \emph{rank-based approach} from \cite{kisfaludi2022gap}.

Let $\Gamma_B$ denote the set of all perfect matchings on $B$. We say that $M_1$ and $M_2$ in $\Gamma_B$ \emph{fit} if their union is a Hamiltonian Cycle on $B$. 

\begin{definition}[representation, \cite{kisfaludi2022gap}]
Let $B$ be a set. Let $\mathcal{A}$ and $\mathcal{A}'$ be two subsets of $\Gamma_B\times \mathbb{N}$.
We say that $\mathcal{A'}$ \emph{represents} $\mathcal{A}$ if for all $M\in\Gamma_B$ we have \[\max\{\kappa | (M',\kappa)\in\mathcal{A'}\text{ and }M\text{ fits }M'\}=\max\{\kappa | (M',\kappa)\in\mathcal{A}\text{ and }M\text{ fits }M'\}.\]
\end{definition}


\begin{lemma}[reduce, {\cite[Theorem~3.7]{bodlaender2015deterministic}}, see also {\cite[Lemma~5.2 in the full version]{kisfaludi2022gap}}]
\label{lem:reduce}
There exists an algorithm, called \emph{\texttt{reduce}}, that given a set $B$ and $\mathcal{A}\subseteq\Gamma_B\times\mathbb{N}$, computes in time $|\mathcal{A}|\cdot 2^{O(|B|)}$ a set $\mathcal{A}'\subseteq\mathcal{A}$ such that $\mathcal{A}'$ represents $\mathcal{A}$ and $|\mathcal{A}'|\leq 2^{|B|-1}$.
\end{lemma}

Let $C$ be a cell in the quadtree and let $s\in \mathcal{S}$.
We define the family $\{\RR^C_s[B]\}_B$ and the set $\RR^C_s$ in the same way as we define the family $\{\MM^C_s[B]\}_B$ and the set $\MM^C_s$ in \cref{sec:first-algo}, except that we use \texttt{reduce} so as to keep the number of elements in $\RR_s^C[B]$ bounded.

\begin{remark}
    It is standard to enrich the dynamic program so that we obtain a collection of paths instead of the total length of that collection. Indeed, once the dynamic programming table is computed, one can recursively reconstruct the corresponding path.
\end{remark}

\subsubsection{Analysis}

\begin{lemma}[adaptation from {\cite[Lemma 5.3 in the full version]{kisfaludi2022gap}}]
\label{lem:repr}
For any cell $C$ in the quadtree, any budget $s\in \mathcal{S}$, and any fine multiset $B\subseteq\partial C$, the set $\RR^C_s[B]$ represents $\MM^C_s[B]$.
\end{lemma}

\cref{lem:RT} is an adaptation from \cite{kisfaludi2022gap}. 

\begin{lemma}[adaptation from {\cite[Lemma~5.4 and Claim~5.5 in the full version]{kisfaludi2022gap}}]
\label{lem:RT}
The running time of the algorithm for all cells $C$ in the quadtree, for all budgets $s\in\mathcal{S}$, and for all fine multisets $B\subseteq \partial C$ is $n\cdot 2^{O(r^{d-1})}\cdot \log^{2d^2\cdot 2^d} n$.
\end{lemma}

\begin{proof}[Proof of \cref{thm:DP}]
From the structure theorem (\cref{thm:structure}), with probability at least $1/2$, there exists a $k$-salesman tour $\pi$ that is $(m,r)$-simple and such that
\begin{equation}
\label{eq:struct}
    w(\pi)\leq (1+\eps)\cdot\opt.
\end{equation}

We condition on the above event in the rest of the analysis.

According to \cite{arora1998polynomial}, \[\opt\leq dk^{1-1/d}L=\Phi,\]
where the equality follows by the definition of $\Phi$ (\cref{def:budget}). 
Therefore, 
\[w(\pi)\leq (1+\eps)\cdot \Phi.\]
Let $C_0$ be the root cell of the quadtree. Since $\pi$ is a closed path strictly contained in $C_0$, the set of points of $\partial C_0$ is $B:=\emptyset$, and the only matching on $B$ is $M:=\emptyset$. 
Let $\kappa$ be the number of points visited by $\pi$. 
Since $\pi$ is a $k$-salesman tour, $\kappa\geq k$.
Since $w(\pi)\leq (1+\eps)\cdot \Phi$, we may apply \cref{lem:discrete} on $C_0$ and $\{\pi\}$ and obtain an integer $\kappa'\geq \kappa$ and a budget $s\in \mathcal{S}$ such that $(\emptyset, \kappa')\in\MM_s^{C_0}$ and 
\begin{equation}
\label{eq:repr}
    s\leq (1+\eps)\cdot w(\pi).
\end{equation}
\cref{lem:repr} ensures that $\RR_s^{C_0}[\emptyset]$ represents $\MM_s^{C_0}[\emptyset]$, hence
\[\max\{\kappa''|(M', \kappa'')\in\MM_s^{C_0}[\emptyset]\text{ and }\emptyset\text{ fits } M'\}=\max\{\kappa''|(M', \kappa'')\in\RR_s^{C_0}[\emptyset]\text{ and }\emptyset\text{ fits } M'\}.\]

Since $\emptyset$ fits $\emptyset$ and $(\emptyset, \kappa')\in\MM_s^{C_0}[\emptyset]$,  we have $\kappa'\leq \max\{\kappa''|(M', \kappa'')\in\MM_s^{C_0}[\emptyset]\text{ and }\emptyset\text{ fits } M'\}$. Therefore, there exists $(M', \kappa'')\in \RR_s^{C_0}[\emptyset]$ such that $M'$ fits $\emptyset$ and $\kappa''\geq \kappa'$.
The only matching $M'$ on $\emptyset$ that fits $\emptyset$ is $\emptyset$, hence $M'=\emptyset$. Thus, $(\emptyset, \kappa'')\in\RR_{s}^{C_0}$, for some $\kappa''\geq\kappa'\geq k$.

Let $s^*$ be the output of the algorithm, which is the minimum budget such that $(\emptyset, \kappa'')\in\RR_{s^*}^{C_0}$ for some $\kappa''\geq k$. From \eqref{eq:struct} and \eqref{eq:repr}, we have
\[s^*\leq (1+3\eps)\cdot \opt.\]
Replacing $\eps$ by $\eps':=\eps/3$ leads to the approximation ratio in the claim.

The running time in the claim follows from \cref{lem:RT}.

For the derandomization, observe that the only step using randomness is the random shift to construct the quadtree.
Since there are $O(n^d)$ possible shifts, the  algorithm can be derandomized by increasing the running time by a factor $O(n^d)$.

This completes the proof of \cref{thm:DP}.
\end{proof}

\begin{remark}
The spanner techniques introduced by Rao and Smith \cite{rao1998approximating} lead to a better running time for TSP, but those techniques do not seem to apply to $k$-TSP. Indeed, a key property for TSP is the existence of a near-optimal solution using the spanner only (see Lemma 5.1 of \cite{kisfaludi2022gap}). However, this property does not hold for $k$-TSP, since the solution to $k$-TSP might be much less expensive compared with the spanner of the entire graph. 
\end{remark}

\bibliographystyle{plainurl}
\bibliography{biblio}

\input{appendix}

\end{document}

%% file: appendix.tex
\appendix 
\newpage
\section{Pseudocode for the First Algorithm}
\label{sec:bmp-first}
\begin{algorithm}[H]
\caption{$\texttt{Budget-Multipath-First}(C, s, B)$}\label{alg:bmp-first}
\If{$C$ is a leaf cell}{
    \If{$|B|=0$}{
        $\MM^C_s[B]\gets \{(\emptyset, 0)\}$
    }
    \ElseIf{$|B|=2$}{
        Let $u$ and $v$ be the two points in $B$\\
        \If{$C\cap P\neq \emptyset$}{
            Let $p$ be the (only) point in $C\cap P$ and let $n_p$ be the multiplicity of $p$ in $P$\\
            \If{$\text{dist}(u, p)+\text{dist}(p, v)\leq s$}{
                $\MM^C_s[B]\gets\{(\{(u,v)\}, n_p)\}$
            }
            \ElseIf{$\text{dist}(u, v)\leq s$}{
                $\MM^C_s[B]\gets\{(\{(u,v)\}, 0)\}$
            }
            \Else{
                $\MM^C_s[B]\gets\emptyset$
            }
        }
        \Else{
            \If{$\text{dist}(u, v)\leq s$}{
                $\MM^C_s[B]\gets\{(\{(u,v)\}, 0)\}$
            }
            \Else{
                $\MM^C_s[B]\gets\emptyset$
            }
        }
    }
    \Else{
        \For{$s_1,s_2\in\mathcal{S}$ such that $s_1+s_2\leq s$}{
            \For{$u, v\in B$ such that $u\neq v$}{
                \For{$(M,\kappa)\in\MM^C_{s_1}[B\setminus\{u,v\}]$}{
                    \If{$\text{dist}(u,v)\leq s_2$}{
                        Insert $(M\cup\{(u,v)\},\kappa)$ into $\MM^C_s[B]$
                    }
                }
            }
        }
    }
}
\Else{
    Let $C_1,\ldots,C_{\powerd}$ be the children of $C$\\
    \For{$s_1,\dots,s_{\powerd}\in\mathcal{S}$ such that $\sum_i s_i\leq s$}{
        \For{$(M_1,\kappa_1)\in\MM^{C_1}_{s_1},\ldots,(M_{\powerd},\kappa_{\powerd})\in\MM^{C_{\powerd}}_{s_{\powerd}}$}{
            \If{$M_1,\ldots,M_{\powerd}$ are compatible}{
                Let $M\gets \texttt{Join}(M_1,\dots,M_{\powerd})$\\
                \If{the set of endpoints of $M$ equals $B$}{
                    Insert $(M,\sum_i \kappa_i)$ into $\MM^C_s[B]$
                }
            }
        }
    }
}
\end{algorithm}

\newpage
\section{Pseudocode for the Improved Algorithm}

The difference with \cref{alg:bmp-first} is on the last line of the \texttt{if} and \texttt{else} blocks, where \texttt{reduce} was added as explained in \cref{sec:improved-algo}.

\label{sec:bmp-improved}
\begin{algorithm}[H]
\caption{$\texttt{Budget-Multipath-Improved}(C, s, B)$}\label{alg:bmp-improved}
\If{$C$ is a leaf cell}{
    \If{$|B|=0$}{
        $\RR^C_s[B]\gets \{(\emptyset, 0)\}$
    }
    \ElseIf{$|B|=2$}{
        Let $u$ and $v$ be the two points in $B$\\
        \If{$C\cap P\neq \emptyset$}{
            Let $p$ be the (only) point in $C\cap P$ and let $n_p$ be the multiplicity of $p$ in $P$\\
            \If{$\text{dist}(u, p)+\text{dist}(p, v)\leq s$}{
                $\RR^C_s[B]\gets\{(\{(u,v)\}, n_p)\}$
            }
            \ElseIf{$\text{dist}(u, v)\leq s$}{
                $\RR^C_s[B]\gets\{(\{(u,v)\}, 0)\}$
            }
            \Else{
                $\RR^C_s[B]\gets\emptyset$
            }
        }
        \Else{
            \If{$\text{dist}(u, v)\leq s$}{
                $\RR^C_s[B]\gets\{(\{(u,v)\}, 0)\}$
            }
            \Else{
                $\RR^C_s[B]\gets\emptyset$
            }
        }
    }
    \Else{
        \For{$s_1,s_2\in\mathcal{S}$ such that $s_1+s_2\leq s$}{
            \For{$u, v\in B$ such that $u\neq v$}{
                \For{$(M,\kappa)\in\RR^C_{s_1}[B\setminus\{u,v\}]$}{
                    \If{$\text{dist}(u,v)\leq s_2$}{
                        Insert $(M\cup\{(u,v)\},\kappa)$ into $\RR^C_s[B]$
                    }
                }
            }
        }
        $\RR^C_s[B]\gets \texttt{reduce}(\RR^C_s[B])$
    }
}
\Else{
    Let $C_1,\ldots,C_{\powerd}$ be the children of $C$\\
    \For{$s_1,\dots,s_{\powerd}\in\mathcal{S}$ such that $\sum_i s_i\leq s$}{
        \For{$(M_1,\kappa_1)\in\RR^{C_1}_{s_1},\ldots,(M_{\powerd},\kappa_{\powerd})\in\RR^{C_{\powerd}}_{s_{\powerd}}$}{
            \If{$M_1,\ldots,M_{\powerd}$ are compatible}{
                Let $M\gets \texttt{Join}(M_1,\dots,M_{\powerd})$\\
                \If{the set of endpoints of $M$ equals $B$}{
                    Insert $(M,\sum_i \kappa_i)$ into $\RR^C_s[B]$
                }
            }
        }
    }
    $\RR^C_s[B]\gets\texttt{reduce}(\RR^C_s[B])$
}
\end{algorithm}

\newpage

\section{Proof of \cref{lem:RT}}
\label{sec:proof-lem:RT}
Let $R_{\text{max}}$ be the maximum size of any $\RR^C_s$.
Let $B_{\text{max}}$ be the maximum size of any fine multiset $B\subseteq \partial C$ over all cells $C$.
Let $F_{\text{max}}$ be the maximum number different fine multisets $B\subseteq\partial C$ over all cells $C$.

Fix a cell $C$, a budget $s\in\mathcal{S}$, and a fine multiset $B\subseteq \partial C$. We analyze the running time of \cref{alg:bmp-improved} on $(C,s,B)$. If $C$ is a leaf cell, there are at most $O(|\mathcal{S}|^2\cdot B_{\text{max}}^2\cdot R_{\text{max}})$ insert operations. Since $\RR_s^C[B]$ has size at most $O(|\mathcal{S}|^2\cdot B_{\text{max}}^2\cdot R_{\text{max}})$, by \cref{lem:reduce}, the call to \texttt{reduce} runs in time $O(|\mathcal{S}|^2\cdot B_{\text{max}}^2\cdot R_{\text{max}}\cdot 2^{B_\text{max}})$. Therefore, the running time is $O(|\mathcal{S}|^2\cdot B_{\text{max}}^2\cdot R_{\text{max}}\cdot 2^{B_\text{max}})$. If $C$ is not a leaf cell, there are at most $O(|\mathcal{S}|^\powerd \cdot R_\text{max}^\powerd)$ insert operations. Since $\RR_s^C[B]$ has size at most $O(|\mathcal{S}|^\powerd \cdot R_\text{max}^\powerd)$, by \cref{lem:reduce}, the call to \texttt{reduce} runs in time $O(|\mathcal{S}|^\powerd \cdot R_\text{max}^\powerd\cdot 2^{O(B_\text{max})})$. Thus, the running time is $O(|\mathcal{S}|^\powerd \cdot R_\text{max}^\powerd\cdot 2^{O(B_\text{max})})$. Hence, in both cases, the running time of \cref{alg:bmp-improved} is $O(|\mathcal{S}|^\powerd\cdot R_\text{max}^\powerd\cdot 2^{O(B_\text{max})})$.

From \cite{arora1998polynomial}, there are $O(n\log n)$ cells in the quadtree. Hence there are $O(n\log n\cdot |\mathcal{S}|\cdot F_\text{max})$ tuples $(C, s, B)$. Therefore, the running time of \cref{alg:bmp-improved} for all tuples $(C,s,B)$ is $O(n\log n\cdot F_\text{max}\cdot |\mathcal{S}|^{\powerd+1}\cdot R_\text{max}^\powerd\cdot 2^{O(B_\text{max})})$.

To analyze that running time, we first bound $|\mathcal{S}|$.

\begin{fact}
    $|\mathcal{S}|=O(\log^2 n/\eps)$.
\end{fact}
\begin{proof}
Let $L$ be the side-length of the bounding box. The budgets are in $\{0\}\cup[1/(r^2+m^{1/(d-1)}), (1+\eps)^2\cdot \Phi]$, where $\Phi=dk^{1-1/d}L$ (\cref{def:budget}). Since the instance is well-rounded, $L=O(k^2)=O(n^2)$ (\cref{def:well-rounded}), therefore there are at most $O(\log^2 n/\eps)$ different budgets. Hence $|\mathcal{S}|=O(\log^2 n/\eps)$.
\end{proof}

By \cref{fct:finebound}, we have
\[B_\text{max}\leq 2d\cdot 2r^{d-1}.\]

Next, we bound $F_\text{max}$. Let $C$ be a cell and $F$ be one of its facets.

For any fine multiset $B\subseteq\partial C$, let $b_F$ denote the number of points of $B$ that are on $F$. By \cref{def:fine}, either one of the following two cases holds: (1) $b_F\leq 1$ and $B\cap F\subseteq\text{grid}(F, m)$; (2) $b_F\geq 2$ and $B\cap F\subseteq\text{grid}(F, g(b_F))$ for some $g(b_F)\leq r^{2d-2}/b_F$. Moreover, each point in $\text{grid}(F, g(b_F))$ occurs at most twice in $B\cap F$.

We show that $b_F\leq 2r^{d-1}$. In case (1), this is trivial. In case (2), since each point from $\text{grid}(F, g(b_F))$ is contained at most twice in $B\cap F$, we have $b_F\leq 2g(b_F)$. Together with $g(b_F)\leq r^{2d-2}/b_F$, we have $b_F\leq 2r^{d-1}$.

In case (1), either $b_F = 0$ or there are $m$ possibilities for the unique point in $B\cap F\subseteq\text{grid}(F,m)$. In case (2), for a fixed $b_F\geq 2$ there are at most ${2g(b_F)\choose b_F}$ possibilities for $B\cap F$ because no point is present more than twice in $B\cap F$. Since ${2g(b_F)\choose b_F}\leq {2r^{2d-2}/b_F\choose b_F}$, the total number of possibilities for $B\cap F$ is at most
\[1+m+\sum_{b_F=2}^{2r^{d-1}}{2r^{2d-2}/b_F\choose b_F}.\]

By \cref{def:fine}, $m= (O((\sqrt{d}/\eps)\log L))^{d-1}$. By \cite[Claim~3.4]{kisfaludi2022gap}, ${2r^{2d-2}/b_F\choose b_F}=2^{O(r^{d-1})}$. Thus, the number of possibilities for $B\cap F$ is at most
\[(O((\sqrt{d}/\eps)\log L))^{d-1} + \sum_{b_F=2}^{2r^{d-1}} 2^{O(r^{d-1})} = (O((\sqrt{d}/\eps)\log L))^{d-1} + 2^{O(r^{d-1})}.\]

\begin{fact}
\label{fact:2t}
    Let $x, y\in\mathbb{R}$ and $t\in\mathbb{N}$. We have $(x+y)^t\leq 2^t\cdot x^t\cdot y^t$.
\end{fact}

By \cref{fact:2t} and the fact that $C$ has $2d$ facets,
\begin{align*}
F_\text{max} &= \left((O((\sqrt{d}/\eps)\log L))^{d-1} + 2^{O(r^{d-1})}\right)^{2d}\\
&\leq 2^{2d}\cdot (O((\sqrt{d}/\eps)\log L))^{(d-1)2d}\cdot \left(2^{O(r^{d-1})}\right)^{2d}.
\end{align*}

Moreover, since $d$ is a constant and $L=O(n^2)$, we have
\[F_\text{max}=2^{O(r^{d-1})}\cdot \log^{(d-1)\cdot 2d} n.\]

Finally, we bound $R_\text{max}$. From \cref{lem:reduce}, since $\texttt{reduce}$ is applied to each $\RR_s^C[B]$, we have $|\RR_s^C[B]|\leq 2^{B_\text{max}-1}$. Thus
\[R_\text{max}\leq F_\text{max}\cdot 2^{B_\text{max}-1}\leq 2^{O(r^{d-1})}\cdot \log^{(d-1)\cdot 2d} n\cdot 2^{2d\cdot 2r^{d-1} - 1} = 2^{O(r^{d-1})}\cdot \log^{(d-1)\cdot 2d} n.\]



In conclusion,
\begin{align*}
n\log n\cdot F_\text{max}\cdot |\mathcal{S}|^{\powerd+1}\cdot R_\text{max}^\powerd\cdot 2^{O(B_\text{max})} &= n\cdot \left(2^{O(r^{d-1})}\right)^{1+\powerd}\cdot \log^{1 + (d-1)\cdot 2d + 2\powerd + (d-1)\cdot 2d\cdot\powerd} n\\
& \;\;\;\;\;\;\;\; \cdot (1/\eps)^\powerd\cdot 2^{O(2d\cdot 2r^{d-1})}\\
&\leq n\cdot 2^{O(r^{d-1})}\cdot \log^{2d^2\cdot 2^d} n.
\end{align*}

This completes the proof of the claim.